\documentclass[letterpaper, 10 pt, conference]{ieeeconf}  

\IEEEoverridecommandlockouts                              
\usepackage{amsmath,amsfonts,amssymb,color}
\usepackage{graphicx}
\usepackage{psfrag}
\usepackage{algorithm}
\usepackage{algpseudocode}
\usepackage{array}
\usepackage{cite}
\usepackage{footmisc}
\usepackage{mathtools}
\usepackage{enumerate}
\usepackage{soul,color}
\usepackage{caption}
\usepackage{subfig}

\newtheorem{theorem}{Theorem}

\newtheorem{lemma}{Lemma}
\newtheorem{remark}{Remark}
\newtheorem{definition}{Definition}
\newtheorem{proposition}{Proposition}
\newtheorem{example}{Example}
\newtheorem{assumption}{Assumption}

\renewcommand{\theenumi}{(\alph{enumi}}
\newtheorem{Remark}{Remark}

\DeclareMathOperator{\vect}{vec}

\makeatletter
\newcommand{\distas}[1]{\mathbin{\overset{#1}{\kern\z@\sim}}}%
\newsavebox{\mybox}\newsavebox{\mysim}
\newcommand{\distras}[1]{%
  \savebox{\mybox}{\hbox{\kern3pt$\scriptstyle#1$\kern3pt}}%
  \savebox{\mysim}{\hbox{$\sim$}}%
  \mathbin{\overset{#1}{\kern\z@\resizebox{\wd\mybox}{\ht\mysim}{$\sim$}}}%
}
\makeatother

\title{\LARGE \bf
Finite Sample Guarantees for Distributed Online Parameter Estimation with Communication Costs
}

\author{~Lei~Xin, George Chiu, Shreyas Sundaram 
\thanks{This research was supported by USDA grant 2018-67007-28439.  This work represents the opinions of the authors and not the USDA or NIFA. Lei Xin and Shreyas Sundaram are with the Elmore Family School of Electrical and Computer Engineering, Purdue University. George Chiu is with the School of Mechanical Engineering, Purdue University. E-mails: {\tt\{lxin, gchiu, sundara2\}@purdue.edu}.
}
}

\begin{document}

\maketitle
\thispagestyle{empty}
\pagestyle{empty}

\begin{abstract}
We study the problem of estimating an unknown parameter in a distributed and online manner. Existing work on distributed online learning typically either focuses on asymptotic analysis, or provides bounds on regret. However, these results may not directly translate into bounds on the error of the learned model after a finite number of time-steps. In this paper, we propose a distributed online estimation algorithm which enables each agent in a network to improve its estimation accuracy by communicating with neighbors. We provide non-asymptotic bounds on the estimation error, leveraging the statistical properties of the underlying model. Our analysis demonstrates a trade-off between estimation error and communication costs. Further, our analysis allows us to determine a time at which the communication can be stopped (due to the costs associated with communications), while meeting a desired estimation accuracy. We also provide a numerical example to validate our results.
\end{abstract}


\section{Introduction} \label{sec: introduction}
Learning an accurate model from data is an important problem in many fields \cite{abu2012learning}, including machine learning, economics, and control theory. However, in many cases, the available datasets are usually split among multiple agents/learners and come in a streaming manner, which require online processing. Coordination among the various agents to process their data also comes with a communication cost, and thus algorithms must be designed to balance the amount of communication with the speed and accuracy of learning.   

The problem of distributed learning/optimization has been studied extensively over the last few decades \cite{nedic2010constrained,zhu2011distributed}. These papers typically provide theoretical guarantees on the convergence of local solutions to the optimizer of the sum of local functions over the network. When it comes to distributed online parameter estimation, the existing literature typically focuses on proving asymptotic convergence of the estimate to the true value \cite{kar2012distributed,zhang2012distributed}. There is another branch of research on distributed online learning that focuses on providing bounds on regret, which is defined as the difference between the costs generated by the sequence of local decisions and the true optimal costs obtained in hindsight \cite{hosseini2013online,hosseini2016online, yuan2020distributed}. The bound on regret can be used as an appropriate metric to evaluate a proposed algorithm, as a sublinear regret implies that the algorithm performs as well as its centralized counterpart on average (over time). However, it is unclear how such bounds can be translated into the bounds on the accuracy of the learned model after a finite number of time-steps. The paper \cite{su2019finite} studies distributed state estimation problem with finite time convergence guarantee with a fixed observation matrix, and under Byzantine faults. In contrast, we consider the problem where the observation/feature matrix is random, which is often encountered in general machine learning problems.  

In this paper, we propose a distributed online parameter estimation algorithm in a networked setting, which enables each agent to improve its estimation accuracy by communicating with neighbors in the network. Our algorithm can be viewed as an extension of the distributed least squares method in \cite{tron2011distributed} to an online setting. In our algorithm, each agent stores two estimates of the true parameter: one computed purely based on local data and one computed after communicating with neighbors in the network. We provide finite time (or sample) upper bounds on the estimation errors of both of these two estimates, which highlight the role of communication. Our results demonstrate a trade-off between estimation error and communication costs. To balance such a trade-off, we discuss how we can leverage our finite time error bounds to determine a time at which the communication can be stopped (due to the costs associated with maintaining communications), while meeting a desired estimation accuracy. We also provide a numerical example to validate our results. 

\section{Notation}
Vectors are taken to be column vectors unless indicated otherwise. Let $\mathbb{R}$ and $\mathbb{Z}$ denote the set of real numbers and integers, respectively. Let $\mathbf{1}_{n}$ denote a vector of dimension $n$ with all of its elements equal to 1. Let $\sigma_{min}(\cdot)$ and $\lambda_{min}(\cdot)$ be the smallest singular value and eigenvalue in magnitude, respectively, of a given matrix. The eigenvalues of a given matrix are ordered with
nonincreasing magnitude, i.e., $|\lambda_{1}(\cdot)|\geq \cdots\geq |\lambda_{min}(\cdot)|$. For a given matrix $A$, we use $A(i,j)$ to denote the element in its $i$-th row and $j$-th column, $A^{*}$ to denote its conjugate transpose, $A^\dagger$ to denote its pseudoinverse, and $\vect(A)$ to denote its vectorization (i.e., the vector obtained by stacking the columns of $A$ starting from the left). We use $\|A\|$, $\|A\|_{1}$ and $\|A\|_{F}$ to denote the spectral norm, $1$-norm. and Frobenius norm, respectively, of matrix $A$. We use $I_{n}$ to denote the identity matrix with dimension $n$. An $n$-dimensional Gaussian distributed random vector is denoted as $u\sim \mathcal{N}(\mu,\sigma^{2} I_{n})$, where $\mu$ is the mean and $\sigma^{2} I_{n}$ is the covariance matrix. The symbol $\cap$ is used to denote the intersection of sets. We use the symbol $\bmod$ to denote the modulo operation. 

\section{Problem formulation} \label{sec: problem formulation}

Consider a group of $m$ agents $\mathcal{V}$ interconnected over an undirected and connected  graph $\mathcal{G} = (\mathcal{V}, \mathcal{E})$. An edge $(i,j)\in \mathcal{E}$ is an unordered pair, which indicates a bidirectional communication link between agents $i$ and $j$. Let $\mathcal{N}_{i}\triangleq \{j:(i,j)\in \mathcal{E}\}$ be the set of neighbors of agent $i$. The goal of these agents is to collaboratively estimate an unknown parameter $\Theta\in \mathbb{R}^{l\times n}$ with finite time guarantees, under a {\it finite number of communication steps}. At each time step $t=1,2,\cdots$, each agent $i\in\mathcal{V}$ gathers the data pair $(x_{i,t},y_{i,t})$ generated by the following model
\begin{equation}
\begin{aligned} 
y_{i,t}=\Theta x_{i,t}+\eta_{i,t}, \\
\end{aligned}
\label{eq:True system}
\end{equation}
where $y_{i,t} \in \mathbb{R}^{l}$ is the label vector, $x_{i,t}\in \mathbb{R}^{n}$ is the feature vector, and $\eta_{i,t}\in\mathbb{R}^{l}$ is the noise. We make the following assumption. 

\begin{assumption} \label{ass:distribution}
The feature vector $x_{i,t}$ and noise $\eta_{i,t}$ are Gaussian random vectors that are independent over time and agents, where $x_{i,t} \sim \mathcal{N}(\mu_{i,t},\sigma_{x}^2I_{n})$ and $\eta_{i,t} \sim \mathcal{N}(0,\sigma_{\eta}^2I_{l})$. The mean $\mu_{i,t}\in\mathbb{R}^{n}$ is deterministic with $\sup\{\|\mu_{i,t}\|:i\in\mathcal{V}, t\in\mathbb{Z}_{\ge 1} \}=\hat{\mu}\in\mathbb{R}_{\geq0}$.
\end{assumption}

The above model can be used to capture many problems. For example, it can be used to capture the problem of dynamical system identification via multiple independent trajectories (assuming zero initial condition and without process noise), where $y_{i,t}$ is the output of the system in each trajectory, $x_{i,t}$ is the input applied in each trajectory, and $\Theta$ is the Markov parameter matrix of the system, e.g., \cite{sun2020finite}. We note that the $x_{i,t}$ considered in our model allows for time-varying and agent-dependent mean $\mu_{i,t}$, and hence is more general than the analogous system identification problem, which typically considers zero-mean Gaussian inputs. We also note that our algorithm does not require any parameters of the model to be known in advance. However, we assume that there are known upper bounds on $\sigma_{x},\sigma_{\eta}, \hat{\mu}, \|\Theta\|$, and there is a known non-zero lower bound on $\sigma_{x}$. These bounds will facilitate the design of certain user-specified parameters in our algorithm, which will become clear when we present our results.   

\begin{Remark}
One may observe that a trivial solution to the above problem might be to not communicate at all, i.e., each agent only updates based on its local dataset. However, such a solution does not leverage the distributed nature of the problem, which provides each agent with the  potential to speed up the learning by communicating with the other agents in the network. On the other hand, communications with the other agents should be carefully designed, as information from others might become less useful when each agent already has a good estimate based on the information it has so far. In the sequel, we study a distributed algorithm that leverages the communication network, which allows all agents to learn the model efficiently (when some upper/lower bounds on $\sigma_{x},\sigma_{\eta}, \hat{\mu}, \|\Theta\|$ are available). More specifically, the algorithm allows every agent to hold an estimate with an estimation error comparable to that of the centralized solution throughout time, while saving communication costs. 
\end{Remark}

\section{A Distributed Online Estimation Algorithm} \label{sec:algorithm}
In this section, we describe a two-time-scale distributed algorithm. At each time step $t=1,2,\cdots$, based on its local dataset, each agent $i\in \mathcal{V}$ wishes to solve the following least squares problem:
\begin{equation} \label{local objective}
\begin{aligned}
    \mathop{\min}_{\Tilde{\Theta}\in \mathbb{R}^{l\times n}} \sum_{j=1}^{t}\|y_{i,j}-\Tilde{\Theta} x_{i,j}\|^{2}_{F}.
\end{aligned}
\end{equation}
The least squares local estimate for agent $i$, given its samples collected up to time step $t$, is 
\begin{equation}   \label{optimal_local_solution}
\begin{aligned}
&\hat{\Theta}_{i,t+1}=(\sum_{j=1}^{t}y_{i,j}x_{i,j}^{*})(\sum_{j=1}^{t}x_{i,j}x_{i,j}^{*})^{-1},
\end{aligned}
\end{equation}
assuming the matrix $\sum_{j=1}^{t}x_{i,j}x_{i,j}^{*}$ is invertible.

The above estimate can be updated iteratively with the arrival of new data pair $(x_{i,t},y_{i,t})$, through
\begin{equation} \label{local update}
\begin{aligned}
&\alpha_{i,t+1}=\alpha_{i,t}+y_{i,t}x_{i,t}^{*},\\
&\beta_{i,t+1}=\beta_{i,t}+x_{i,t}x_{i,t}^{*}, \\
&\hat{\Theta}_{i,t+1}=\alpha_{i,t+1}\beta_{i,t+1}^{\dagger },\\
\end{aligned}
\end{equation}
where $\alpha_{i,1}=0,\beta_{i,1}=0$. Note that $\beta_{i,t+1}^{\dagger}=\beta_{i,t+1}^{-1}$ once $\beta_{i,t+1}$ becomes invertible. Also, $\beta_{i,t+1}^{-1}$ can be updated iteratively using the Sherman-Morrison formula \cite{henderson1981deriving}, which states $\beta_{i,t+1}^{-1}=\beta_{i,t}^{-1}-\frac{\beta_{i,t}^{-1}x_{i,t}x_{i,t}^{*}\beta_{i,t}^{-1}}{1+x_{i,t}^{*}\beta_{i,t}^{-1}x_{i,t}}$.

The algorithm enters the communication phase when the conditions $t\bmod \zeta=0$ and $t\leq S$ are satisfied, where $\zeta\in \mathbb{Z}_{\ge 1}$ and $S\in  \mathbb{Z}_{\ge 0}$, i.e., when the current time step $t$ is an integer multiple of the pre-specified communication period $\zeta$ and is less than the pre-specified stopping time $S$. Letting the superscript $k$ denote the current communication time step, each agent $i\in\mathcal{V}$ sets $\alpha_{i,t+1}^{0}=\alpha_{i,t+1},\beta_{i,t+1}^{0}=\beta_{i,t+1}$. At each communication time step $k$, each agent $i\in\mathcal{V}$ broadcasts its current $\alpha_{i,t+1}^{k}$ and $\beta_{i,t+1}^{k}$ to its neighbors $j\in \mathcal{N}_{i}$, and receives $\alpha_{j,t+1}^{k}$ and $\beta_{j,t+1}^{k}$ from $j\in\mathcal{N}_{i}$. The update is given by
\begin{equation} \label{communication}
\begin{aligned}
&\alpha_{i,t+1}^{k+1}=W(i,i)\alpha_{i,t+1}^{k}+\sum_{j\in\mathcal{N}_{i}}W(i,j)\alpha_{j,t+1}^{k},\\
&\beta_{i,t+1}^{k+1}=W(i,i)\beta_{i,t+1}^{k}+\sum_{j\in\mathcal{N}_{i}}W(i,j)\beta_{j,t+1}^{k},\\
\end{aligned}
\end{equation}
for $k=0, 1, \ldots, T-1$, where $T\in \mathbb{Z}_{\geq 1}$ is the number of pre-specified total communication steps whenever the algorithm enters the communication phase, and $W\in \mathbb{R}^{m\times m}$ is the matrix where $W(i,j)$ is the weight agent $i\in \mathcal{V}$ assigns to agent $j\in \mathcal{V}$. We make the following assumption on $W$, which is commonly used, e.g., \cite{ye2020distributed}.
\begin{assumption} \label{ass:topology}
The weight matrix $W\in \mathbb{R}^{m\times m}$ associated
with the communication graph $G = ( \mathcal{V}, \mathcal{E})$ is assumed to satisfy: (1) $W(i,j)\in \mathbb{R}$ and  $W(i,j)\geq 0$ for
all $i,j \in \mathcal{V}$, and $W(i,j)=0$ if $j \not\in \mathcal{N}_{i}$ and $i \neq j$; (2)$W\mathbf{1}_{m}=\mathbf{1}_{m}$; (3) $W=W^{*}$ and (4) $\rho(W)\triangleq\max\{\lambda_{2}(W), -\lambda_{m}(W)\}< 1$.
\end{assumption}

The local estimate after communication is set to be $\bar{\Theta}_{i,t+1}=\alpha_{i,t+1}^{T}(\beta_{i,t+1}^{T})^{\dagger}$. If there is no communication happened at the current time step $t$, agent $i$ just keeps its estimate from the previous time-step, i.e., $\bar{\Theta}_{i,t+1}=\bar{\Theta}_{i,t}$.

The above steps are encapsulated in Algorithm \ref{algo1}.

\begin{algorithm}
\caption{Distributed Online Estimation Algorithm}
\textbf{Input} Weight matrix $W$, stopping time $S$, communication period $\zeta$, number of communication steps $T$
\label{algo1}
\begin{algorithmic}[1]
\State Each $v_i \in \mathcal{V}$ initializes $\alpha_{i,1}=0,\beta_{i,1}=0, \bar{\Theta}_{i,1}=0$
\For {$t = 1, 2, 3, \ldots$}  
\For {$v_i \in \mathcal V$}  \Comment{Implement in parallel}
\State Gather the data pair $(x_{i,t},y_{i,t})$, where $x_{i,t} \sim \mathcal{N}(\mu_{i,t},\sigma_{x}^{2}I_{n})$
\State Update 
 $\alpha_{i,t+1}, \beta_{i,t+1}, \hat{\Theta}_{i,t+1}$ as in \eqref{local update}
\If{$t\bmod{\zeta}=0$ and $t\leq S$}
\State Set $\alpha_{i,t+1}^{0}=\alpha_{i,t+1},\beta_{i,t+1}^{0}=\beta_{i,t+1}$ 
\For {$k=0, 1, \ldots, T-1$} 
\State Broadcast $\alpha_{i,t+1}^{k}$ $\beta_{i,t+1}^{k}$ to $j\in\mathcal{N}_{i}$, and receive $\alpha_{j,t+1}^{k}$ $\beta_{j,t+1}^{k}$ from $j\in\mathcal{N}_{i}$
\State Update $\alpha_{i,t+1}^{k+1}, \beta_{i,t+1}^{k+1}$ as in \eqref{communication}
\EndFor
\State $\bar{\Theta}_{i,t+1}=\alpha_{i,t+1}^{T}(\beta_{i,t+1}^{T})^{\dagger}$
\Else
\State $\bar{\Theta}_{i,t+1}=\bar{\Theta}_{i,t}$
\EndIf
\EndFor
\EndFor
\end{algorithmic}
\end{algorithm}

\begin{Remark}
Note that Algorithm \ref{algo1} has two time scales. In practice, this captures the scenario where agents can communication multiple times between receiving data samples. Further, note that both $\hat{\Theta}_{i,t+1}$ (without communication) and $\bar{\Theta}_{i,t+1}$ (after communication) are estimates of the true parameter $\Theta$. In the next section, we will provide bounds on the finite time estimation errors $\|\hat{\Theta}_{i,t+1}-\Theta\|$ and $\|\bar{\Theta}_{i,t+1}-\Theta\|$. In practice, one could choose the estimate with smaller (estimated) error bound as the ``true" output of the algorithm. In section \ref{discussion}, we will discuss how to choose the user-specified parameters $\zeta$, $S$ and $T$ to enable efficient learning.
\end{Remark}

\section{Analysis of the Error}
\subsection{Local Estimation Error Without Communication}
We will start with bounding the estimation error using only local samples. Note that for any agent $i\in \mathcal{V}$, we have
\begin{equation} \label{error expression}
\begin{aligned}
\|\hat{\Theta}_{i,t+1}-\Theta\|&=\|\alpha_{i,t+1}\beta_{i,t+1}^{-1}-\Theta\|\\
&=\|(\sum_{j=1}^{t}y_{i,j}x_{i,j}^{*})(\sum_{j=1}^{t}x_{i,j}x_{i,j}^{*})^{-1}-\Theta\|\\
&=\|(\sum_{j=1}^{t}\eta_{i,j}x_{i,j}^{*})(\sum_{j=1}^{t}x_{i,j}x_{i,j}^{*})^{-1}\|\\
&\leq\|\sum_{j=1}^{t}\eta_{i,j}x_{i,j}^{*}\|\|(\sum_{j=1}^{t}x_{i,j}x_{i,j}^{*})^{-1}\|,
\end{aligned}
\end{equation}
assuming the the matrix $\sum_{j=1}^{t}x_{i,j}x_{i,j}^{*}$ is invertible. The proof of our error bound follows by upper bounding the above terms separately.

We will rely on the following lemma from \cite[Corollary~5.35]{vershynin2010introduction}, which provides non-asymptotic lower bound and upper bound of a standard Wishart matrix.
\begin{lemma} Let $u_{i}\sim \mathcal{N}(0,I_{n})$, $i=1,\ldots,t$ be i.i.d random vectors. For any fixed $\delta >0$, with probability at least $1-\delta$, we have both of the following inequalities:
\begin{equation*}
   \sqrt{\lambda_{1}(\sum_{i=1}^{t}u_{i}u_{i}^{*}})\leq \sqrt{t}+\sqrt{n}+\sqrt{2\log{\frac{2}{\delta}}},
\end{equation*}
\begin{equation*}
   \sqrt{\lambda_{min}(\sum_{i=1}^{t}u_{i}u_{i}^{*}})\geq \sqrt{t}-\sqrt{n}-\sqrt{2\log{\frac{2}{\delta}}}.
\end{equation*} 

\label{lemma:Bound of unit variance gaussian}
\end{lemma}

We will use the following lemma from \cite[Lemma~A.1]{oymak2018non}.
\begin{lemma} \label{lemma:gayssian times constant}
Let $M\in\mathbb{R}^{m\times n}$ be a matrix with $m\geq n$, and let $\eta\in\mathbb{R}$ be such that $\|M\|\leq \eta$. Let $Z\in\mathbb{R}^{m\times k}$ be a matrix with independent standard normal entries. Then, for any fixed $\delta>0$, with probability at least $1-\delta$, 
\begin{equation*}
  \|M^{*}Z\|\leq \eta (\sqrt{2(n+k)}+\sqrt{2\log{\frac{2}{\delta}}}).
\end{equation*}
\end{lemma}

We have the following result.
\begin{lemma} \label{local x}
Let Assumption \ref{ass:distribution} hold. Fix $\delta >0$ and let $t\geq \max({t_{1},t_{2}})$, where $t_{1}= 8n+16\log\frac{2}{\delta}, t_{2}=(\frac{16\hat{\mu}(\sqrt{4n}+\sqrt{2\log\frac{2}{\delta}})}{\sigma_{x}})^{2}$. For any $i\in \mathcal{V}$, letting $\bar{\mu}_{i,t}=\frac{4}{t\sigma_{x}^{2}}\sum_{j=1}^{t}\mu_{i,j}\mu_{i,j}^{*}$, with probability at least $1-2\delta$, we have both of the following inequalities:
\begin{equation*}
\begin{aligned}
\|\sum_{j=1}^{t}x_{i,j}x_{i,j}^{*}\|&\leq t(\frac{19}{8}\sigma_{x}^{2}+\hat{\mu}^{2}),\\
\lambda_{min}(\sum_{j=1}^{t}x_{i,j}x_{i,j}^{*})&\geq \frac{t\sigma_{x}^{2}}{8}\lambda_{min}(I_{n}+\bar{\mu}_{i,t}).
\end{aligned}
\end{equation*}
\end{lemma}

\begin{proof}
Fixing $i\in\mathcal{V}$, we can rewrite $x_{i,j}=\sigma_{x}u_{i,j}+\mu_{i,j}$, where $u_{i,j}\sim \mathcal{N}(0,I_{n})$ for $j=1,\ldots.t$. We have
\begin{equation} 
\begin{aligned} 
\sum_{j=1}^{t}x_{i,j}x_{i,j}^{*}&=\sum_{j=1}^{t}(\sigma_{x}u_{i,j}+\mu_{i,j})(\sigma_{x}u_{i,j}^{*}+\mu^{*}_{i,j})\\ 
&=\sum_{j=1}^{t}\sigma_{x}^{2}u_{i,j}u_{i,j}^{*}+\sum_{j=1}^{t}\mu_{i,j}\mu_{i,j}^{*}\\  
&+\sum_{j=1}^{t}\sigma_{x}u_{i,j}\mu_{i,j}^{*}+\sum_{j=1}^{t}\sigma_{x}\mu_{i,j} u_{i,j}^{*}. \label{4 terms}
\end{aligned}
\end{equation}

To derive the upper bound in Lemma \ref{local x}, we start with upper bounding the norm of the first term in the last equality of \eqref{4 terms}. Fixing $\delta>0$ and applying Lemma \ref{lemma:Bound of unit variance gaussian}, we have with probability at least $1-\delta$,
\begin{align} \label{event tmp}
  \sqrt{\lambda_{1}(\sum_{j=1}^{t}u_{i,j}u_{i,j}^{*}})\leq \sqrt{t}+\sqrt{n}+\sqrt{2\log{\frac{2}{\delta}}}.
\end{align}
Further, we have
\begin{equation*}
\begin{aligned}
  &\frac{1}{2}\sqrt{t}\geq \sqrt{n}+\sqrt{2\log{\frac{2}{\delta}}}\\
   \Longleftrightarrow\quad &\frac{t}{4}\geq(\sqrt{n}+\sqrt{2\log{\frac{2}{\delta}}})^{2}.
   \end{aligned}
\end{equation*}
Noting the inequality $2(a^2+b^2)\geq(a+b)^{2}$, we can write $2(n+2\log{\frac{2}{\delta}})\geq  (\sqrt{n}+\sqrt{2\log{\frac{2}{\delta}}})^2$. Letting $t\geq 8n+16\log{\frac{2}{\delta}}$, one can then show that the following holds with probability at least $1-\delta$:
\begin{equation*}
\begin{aligned}
 \sqrt{\lambda_{1}(\sum_{i=1}^{t}u_{i,j}u_{i,j}^{*}})&\leq \sqrt{t}+\sqrt{n}+\sqrt{2\log{\frac{2}{\delta}}}\leq\frac{3}{2}\sqrt{t}.
   \end{aligned}
\end{equation*}

Consequently, we have with probability at least $1-\delta$,
\begin{align} \label{bound 1}
 \|\sum_{i=1}^{t}\sigma_{x}^{2}u_{i,j}u_{i,j}^{*}\| &\leq\frac{9}{4}\sigma_{x}^{2}t.
   \end{align}

For the second term in the last equality of \eqref{4 terms}, from Assumption \ref{ass:distribution}, we have 
\begin{align} \label{bound 2}
\|\sum_{i=1}^{t}\mu_{i,j}\mu_{i,j}^{*}\|\leq t\hat{\mu}^{2}.
\end{align}

For the last two terms in the last equality of \eqref{4 terms}, since $t\geq n$, we have
\begin{align} 
\|\sum_{j=1}^{t}\sigma_{x}u_{i,j}\mu_{i,j}^{*}\|&=\|\sum_{j=1}^{t}\sigma_{x}\mu_{i,j} u_{i,j}^{*}\|\nonumber\\
&=\sigma_{x}\|\begin{bmatrix}\mu_{i,1} \cdots \mu_{i,t} \end{bmatrix}\begin{bmatrix}u_{i,1}\cdots u_{i,t}  \end{bmatrix}^{*} \|\nonumber\\
&\leq  \sigma_{x}\sqrt{t}\hat{\mu}(\sqrt{4n}+\sqrt{2\log\frac{2}{\delta}}), \label{bound 3}
\end{align}
with probability at least $1-\delta$, where the inequality comes from applying Lemma \ref{lemma:gayssian times constant} and the fact that $\|\begin{bmatrix}
\mu_{i,1}\cdots\mu_{i,t}\\
\end{bmatrix}\|\leq\sqrt{t}\hat{\mu}$.
Combining \eqref{bound 1}, \eqref{bound 2} and \eqref{bound 3} using a union bound, letting $t\geq \max\{8n+16\log{\frac{2}{\delta}},(\frac{16\hat{\mu}(\sqrt{4n}+\sqrt{2\log\frac{2}{\delta}})}{\sigma_{x}})^{2}\}$, we have with probability at least $1-2\delta$,
\begin{equation*}
\begin{aligned}
\|\sum_{j=1}^{t}x_{i,j}x_{i,j}^{*}\| &\leq t(\frac{9}{4}\sigma_{x}^{2}+\hat{\mu}^{2})+2\sigma_{x}\sqrt{t}\hat{\mu}(\sqrt{4n}+\sqrt{2\log\frac{2}{\delta}})\\
&\leq t(\frac{9}{4}\sigma_{x}^{2}+\hat{\mu}^{2}) +t\frac{\sigma_{x}^{2}}{8}= t(\frac{19}{8}\sigma_{x}^{2}+\hat{\mu}^{2}),
\end{aligned}
\end{equation*}
which is of the desired form.

Now we prove the lower bound in Lemma \ref{local x}.  We first lower bound the smallest eigenvalue of first term in the last equality of \eqref{4 terms}. Note that conditioning on the event in \eqref{event tmp}, we also have 
\begin{equation*}
   \sqrt{\lambda_{min}(\sum_{j=1}^{t}u_{i,j}u_{i,j}^{*}})\geq \sqrt{t}-\sqrt{n}-\sqrt{2\log{\frac{2}{\delta}}}
\end{equation*}
from Lemma \ref{lemma:Bound of unit variance gaussian}. Letting $t\geq 8n+16\log{\frac{2}{\delta}}$, we have $ \sqrt{\lambda_{min}(\sum_{j=1}^{t}u_{i,j}u^{*}_{i,j}})\geq \frac{1}{2}\sqrt{t}$, which implies
\begin{align} \label{event tmp2}
 \sum_{i=1}^{t}\sigma_{x}^{2}u_{i,j}u_{i,j}^{*} &\succeq\frac{1}{4}\sigma_{x}^{2}tI_{n}.
   \end{align}
Further, from \eqref{4 terms}, note that 
\begin{equation*}
\begin{aligned}
&\lambda_{min}(\sum_{j=1}^{t}x_{i,j}x_{i,j}^{*})\\
&\geq\lambda_{min}(\sum_{j=1}^{t}\sigma_{x}^{2}u_{i,j}u_{i,j}^{*}+\sum_{j=1}^{t}\mu_{i,j}\mu_{i,j}^{*})-2\|\sum_{j=1}^{t}\sigma_{x}u_{i,j}\mu_{i,j}^{*}\|,
\end{aligned}
\end{equation*}
where the inequality comes from \cite[Theorem~3.3.16(c)]{horn1991topics}.
Now, conditioning on the event in \eqref{bound 3} and the event in \eqref{event tmp2}, denoting $\bar{\mu}_{i,t}=\frac{4}{t\sigma_{x}^{2}}\sum_{j=1}^{t}\mu_{i,j}\mu_{i,j}^{*}$, we have
\begin{equation*}
\begin{aligned}
&\lambda_{min}(\sum_{j=1}^{t}\sigma_{x}^{2}u_{i,j}u_{i,j}^{*}+\sum_{j=1}^{t}\mu_{i,j}\mu_{i,j}^{*})-2\|\sum_{j=1}^{t}\sigma_{x}u_{i,j}\mu_{i,j}^{*}\| \\
& \geq \frac{t\sigma_{x}^{2}}{4}\lambda_{min}(I_{n}+\bar{\mu}_{i,t})-
2\sigma_{x}\sqrt{t}\hat{\mu}(\sqrt{4n}+\sqrt{2\log\frac{2}{\delta}}).
\end{aligned}
\end{equation*}
Hence, when $t\geq (\frac{16\hat{\mu}(\sqrt{4n}+\sqrt{2\log\frac{2}{\delta}})}{\sigma_{x}})^{2}$, we have 
\begin{equation*}
\begin{aligned}
&\frac{t\sigma_{x}^{2}}{4}\lambda_{min}(I_{n}+\bar{\mu}_{i,t})-
2\sigma_{x}\sqrt{t}\hat{\mu}(\sqrt{4n}+\sqrt{2\log\frac{2}{\delta}})\\
&\geq \frac{t\sigma_{x}^{2}}{8}\lambda_{min}(I_{n}+\bar{\mu}_{i,t}).
\end{aligned}
\end{equation*}
\end{proof}

Next, we will use the following lemma from\cite[Lemma~1]{dean2019sample} to bound the contribution from the noise terms.
\begin{lemma}
Let $f_{i}\in\mathbb{R}^{l}$, $g_{i}\in\mathbb{R}^{n}$ be independent random vectors $f_{i}\sim \mathcal{N}(0,\Sigma_{f})$ and $g_{i}\sim \mathcal{N}(0,\Sigma_{g})$, for $i=1,\cdots,t$. Let $t\geq2(n+l)\log{\frac{1}{\delta}}$. For any fixed $ \delta>0$, we have with probability at least $1-\delta$,
\begin{equation*}
   \|\sum_{i=1}^{t}f_{i}g_{i}^{*}\|\leq4\|\Sigma_{f}\|^{\frac{1}{2}}\|\Sigma_{g}\|^{\frac{1}{2}}\sqrt{t(n+l)\log{\frac{9}{\delta}}}.
\end{equation*}
\label{lemma:upper bound two independent gaussian}
\end{lemma}

\begin{lemma} \label{lemma:local noise}
Let Assumption \ref{ass:distribution} hold. Fix $\delta >0$ and let $t\geq t_{3}=2(n+l)\log\frac{1}{\delta}$. For any $i\in \mathcal{V}$, we have with probability at least $1-2\delta$,
\begin{equation*}
\begin{aligned}
\|\sum_{j=1}^{t}\eta_{i,j}x_{i,j}^{*}\|&\leq \sqrt{t}\sigma_{\eta}\left(4\sigma_{x}\sqrt{(n+l)\log{\frac{9}{\delta}}}\right.\\
&\quad\quad\quad\quad\quad\left.+\hat{\mu}(\sqrt{2(l+n)}+\sqrt{2\log\frac{2}{\delta}})\right).
\end{aligned}
\end{equation*}
\end{lemma}

\begin{proof}
Fixing $i\in\mathcal{V}$, we can rewrite $\eta_{i,j}=\sigma_{\eta}f_{i,j}$ and $x_{i,j}=\sigma_{x}g_{i,j}+\mu_{i,j}$, where $f_{i,j}$, $g_{i,j}$ are independent Gaussian random vectors with $f_{i,j}\sim \mathcal{N}(0,I_{l})$ and $g_{i,j}\sim \mathcal{N}(0,I_{n})$,  for $j=1,\ldots,t$. We have
\begin{equation*}
\begin{aligned}
&\|\sum_{j=1}^{t}\eta_{i,j}x_{i,j}^{*}\|=\|\sum_{j=1}^{t}\sigma_{\eta}f_{i,j}(\sigma_{x}g_{i,j}^{*}+\mu_{i,j}^{*})\|\\
&\leq\|\sum_{j=1}^{t}\sigma_{\eta}\sigma_{x}f_{i,j}g_{i,j}^{*}\|+\|\sum_{j=1}^{t}\sigma_{\eta}f_{i,j}\mu_{i,j}^{*}\|.
\end{aligned}
\end{equation*}
Fixing $\delta>0$ and letting $t\geq 2(n+l)\log\frac{1}{\delta}$, applying Lemma \ref{lemma:upper bound two independent gaussian}, we have with probability at least $1-\delta$
\begin{align}  \label{bound 4}
&\|\sum_{j=1}^{t}\sigma_{\eta}\sigma_{x}f_{i,j}g_{i,j}^{*}\|\leq 4\sigma_{x}\sigma_{\eta}.\sqrt{t(n+l)\log\frac{9}{\delta}}.
\end{align}

Next, notice that 
\begin{equation*}
\begin{aligned}
\|\sum_{j=1}^{t}\sigma_{\eta}f_{i,j}\mu_{i,j}^{*} \|=\sigma_{\eta}\|\begin{bmatrix}\mu_{i,1} \cdots \mu_{i,t} \end{bmatrix}\begin{bmatrix}f_{i,1}\cdots f_{i,t}  \end{bmatrix}^{*} \|.
\end{aligned}
\end{equation*}
Using the fact that $\|\begin{bmatrix}
\mu_{i,1}\cdots\mu_{i,t}\\ \end{bmatrix}\|\leq\sqrt{t}\hat{\mu}$ and applying Lemma \ref{lemma:gayssian times constant}, we have with probability at least $1-\delta$
\begin{align}    \label{bound 5}
&\|\sum_{j=1}^{t}\sigma_{\eta}f_{i,j}\mu_{i,j}^{*}\|\leq  \sigma_{\eta}\sqrt{t}\hat{\mu}(\sqrt{2(n+l)}+\sqrt{2\log\frac{2}{\delta}}).
\end{align}
Applying a union bound over the events in  \eqref{bound 4} and \eqref{bound 5}, we get the desired form.
\end{proof}
\begin{theorem} \label{Thm:local bound}
Let Assumption \ref{ass:distribution} hold. Fix $\delta >0$ and let $t\geq \max({t_{1},t_{2},t_{3}})$, where $t_{1}= 8n+16\log\frac{2}{\delta}, t_{2}=(\frac{16\hat{\mu}(\sqrt{4n}+\sqrt{2\log\frac{2}{\delta}})}{\sigma_{x}})^{2},t_{3}=2(n+l)\log\frac{1}{\delta}$. For any $i\in \mathcal{V}$, letting $\bar{\mu}_{i,t}=\frac{4}{t\sigma_{x}^{2}}\sum_{j=1}^{t}\mu_{i,j}\mu_{i,j}^{*}$, we have with probability at least $1-4\delta$,
\begin{equation}
\begin{aligned}
&\|\hat{\Theta}_{i,t+1}-\Theta\|\leq \frac{C_{1}}{\sqrt{t}\sigma_{x}^{2}\lambda_{min}(I_{n}+\bar{\mu}_{i,t})},
\end{aligned}
\end{equation}
where $C_{1}=8\sigma_{\eta}(4\sigma_{x}\sqrt{(n+l)\log\frac{9}{\delta}}+\hat{\mu}(\sqrt{2(l+n)}+2\sqrt{\log\frac{2}{\delta}}))$.
\end{theorem}
\begin{proof}
Recall the expression of the estimation error in \eqref{error expression}. Noting that $\|(\sum_{j=1}^{t}x_{i,j}x_{i,j}^{*})^{-1}\|=\frac{1}{\lambda_{min}(\sum_{j=1}^{t}x_{i,j}x_{i,j}^{*})}$, we can combine the second event in Lemma \ref{local x} and the event in Lemma \ref{lemma:local noise} via a union bound to get the desired result.
\end{proof}
\begin{Remark}
Theorem \ref{Thm:local bound} shows that the error is $\mathcal{O}(\frac{1}{\sqrt{t}})$. Note that when the mean $\mu_{i,j}$ is non-zero but invariant for all $t$, the bound could become more conservative when $n>1$ (since that makes $\hat{\mu}$ in $C_{1}$ larger, but we still have $\lambda_{min}(I_{n}+\bar{\mu}_{i,t})=1$). If this is known in advance, one could define a new pair of sequences $\hat{y}_{i,t}=y_{i,2t-1}-y_{i,2t}$ and $\hat{x}_{i,t}=x_{i,2t-1}-x_{i,2t}$. One then has $\hat{y}_{i,t}=\Theta \hat{x}_{i,t}+\hat{\eta}_{i,t}$, where $\hat{x}_{i,t} \sim \mathcal{N}(0,2\sigma_{x}^2I_{n})$ and $\hat{\eta}_{i,t} \sim \mathcal{N}(0,2\sigma_{\eta}^2I_{l})$. The same bound will still apply to the least squares solution using the transformed dataset, i.e., with the price of reducing the amount of samples by one-half, one could force the mean-dependent terms in Theorem \ref{Thm:local bound} to go to zero. Such a transformation could result in a smaller bound when $\hat{\mu}$ is large enough.
\end{Remark}

\subsection{Global Estimation Error}
Next, we look at the estimation error of the least squares estimate supposing that one has access to all samples across the network up to time step $t$. The global estimate and its associated estimation error are
\begin{equation}   \label{optimal_global_solution}
\begin{aligned}
\hat{\Theta}_{t+1}&\triangleq(\sum_{i=1}^{m}\sum_{j=1}^{t}y_{i,j}x_{i,j}^{*})(\sum_{i=1}^{m}\sum_{j=1}^{t}x_{i,j}x_{i,j}^{*})^{-1}\\
&=(\sum_{i=1}^{m}\alpha_{i,t+1})(\sum_{i=1}^{m}\beta_{i,t+1})^{-1},
\end{aligned}
\end{equation}
\begin{equation*}  
\begin{aligned}
\|\hat{\Theta}_{t+1}-\Theta\|=\|(\sum_{i=1}^{m}\sum_{j=1}^{t}\eta_{i,j}x_{i,j}^{*})(\sum_{i=1}^{m}\sum_{j=1}^{t}x_{i,j}x_{i,j}^{*})^{-1}\|,
\end{aligned}
\end{equation*}
assuming the matrix $\sum_{i=1}^{m}\sum_{j=1}^{t}x_{i,j}x_{i,j}^{*}$ is invertible. The proof of the following theorem entirely follows Theorem \ref{Thm:local bound} due to Assumption \ref{ass:distribution}, with slight adjustments to accommodate possibly different means of $x_{i,j}$ across the network.

\begin{theorem} \label{Thm:global error}
Let Assumption \ref{ass:distribution} hold. Fix $\delta >0$, and let $t\geq \frac{1}{m}\max({t_{1},t_{2},t_{3}})$, where $t_{1}= 8n+16\log\frac{2}{\delta}, t_{2}=(\frac{16\hat{\mu}(\sqrt{4n}+\sqrt{2\log\frac{2}{\delta}})}{\sigma_{x}})^{2},t_{3}=2(n+l)\log\frac{1}{\delta}$. Letting $\bar{\mu}_{t}=\frac{4}{mt\sigma_{x}^{2}}\sum_{i=1}^{m}\sum_{j=1}^{t}\mu_{i,j}\mu_{i,j}^{*}$, we have with probability at least $1-4\delta$,
\begin{equation*}  
\begin{aligned}
&\|\hat{\Theta}_{t+1}-\Theta\|\leq \frac{C_{1}}{\sqrt{mt}\sigma_{x}^{2}\lambda_{min}(I_{n}+\bar{\mu}_{t})},
\end{aligned}
\end{equation*}
where $C_{1}$ is defined in Theorem \ref{Thm:local bound}.
\end{theorem}
\begin{Remark}
Theorem \ref{Thm:global error} indicates that the global estimation error bound is approximately $\frac{1}{\sqrt{m}}$ of the local estimation error bound for agent $i\in\mathcal{V}$ in Theorem \ref{Thm:local bound} (when $\bar{\mu}_{t}$ in Theorem \ref{Thm:global error} is approximately equal to $\bar{\mu}_{i,t}$ in Theorem \ref{Thm:local bound}). Next, we will analyze the local estimation error after finite communication steps, which shows how communication could help agents benefit from the global dataset.  
\end{Remark}

\subsection{Local Estimation Error After Communication}
To derive the error bound of the local estimate after communication, we first define some quantities for notational simplicity. Recall the roles of the stopping time $S$ and the communication period $\zeta$ in Algorithm \ref{algo1}. For $t$ satisfying $t\bmod{\zeta}=0$ and $t\leq S$ (note that we will only consider such $t$ in this section), define $\bar{\alpha}_{t+1}\triangleq \frac{1}{m}\sum_{i=1}^{m}\alpha_{i,t+1}$ and $\bar{\beta}_{t+1} \triangleq \frac{1}{m}\sum_{i=1}^{m}\beta_{i,t+1}$. For any $i\in \mathcal{V}$, note that
\begin{equation}  \label{Communicated error}
\begin{aligned}
&\|\bar{\Theta}_{i,t+1}-\Theta\|=\|\alpha_{i,t+1}^{T}(\beta_{i,t+1}^{T})^{-1}-\Theta\|\\
&=\|\alpha_{i,t+1}^{T}(\beta_{i,t+1}^{T})^{-1}-\bar{\alpha}_{t+1}\bar{\beta}_{t+1}^{-1}+\bar{\alpha}_{t+1}\bar{\beta}_{t+1}^{-1}-\Theta\|\\
&\leq \|\alpha_{i,t+1}^{T}(\beta_{i,t+1}^{T})^{-1}-\bar{\alpha}_{t+1}\bar{\beta}_{t+1}^{-1}\|+\|\bar{\alpha}_{t+1}\bar{\beta}_{t+1}^{-1}-\Theta\|,
\end{aligned}
\end{equation}
under the invertibility assumption.

The second portion of the above inequality can be bounded using Theorem \ref{Thm:global error}, since $\|\bar{\alpha}_{t+1}\bar{\beta}_{t+1}^{-1}-\Theta\|=\|\hat{\Theta}_{t+1}-\Theta\|$. Now we will focus on bounding the first term, which corresponds to the error due to network convergence at time step $t$, using $T$ steps of communication. For $t$ satisfying $t\bmod{\zeta}=0$ and $t\leq S$, fixing $i\in \mathcal{V}$ and defining $\epsilon_{\bar{\alpha}_{i,t+1}^{T}}\triangleq\alpha_{i,t+1}^{T}-\bar{\alpha}_{t+1}$ and $\epsilon_{(\beta_{i,t+1}^{T})^{-1}}\triangleq(\beta_{i,t+1}^{T})^{-1}-\bar{\beta}_{t+1}^{-1}$, we have
\begin{equation}  \label{decomposition of error}
\begin{aligned}
&\|\alpha_{i,t+1}^{T}(\beta_{i,t+1}^{T})^{-1}-\bar{\alpha}_{t+1}\bar{\beta}_{t+1}^{-1}\|=\\
&\|(\bar{\alpha}_{t+1}+\epsilon_{\bar{\alpha}_{i,t+1}^{T}})(\bar{\beta}_{t+1}^{-1}+\epsilon_{(\beta_{i,t+1}^{T})^{-1}})-\bar{\alpha}_{t+1}\bar{\beta}_{t+1}^{-1}\|\\
&\leq \|\bar{\alpha}_{t+1}\epsilon_{(\beta_{i,t+1}^{T})^{-1}}\|+\|\epsilon_{\bar{\alpha}_{i,t+1}^{T}}\bar{\beta}_{t+1}^{-1}\|+\|\epsilon_{\bar{\alpha}_{i,t+1}^{T}}\epsilon_{(\beta_{i,t+1}^{T})^{-1}}\| \\
&\leq\|\bar{\alpha}_{t+1}\|\|\epsilon_{(\beta_{i,t+1}^{T})^{-1}}\|+\|\epsilon_{\bar{\alpha}_{i,t+1}^{T}}\|\|\bar{\beta}_{t+1}^{-1}\|\\
&\qquad +\|\epsilon_{\bar{\alpha}_{i,t+1}^{T}}\|\|\epsilon_{(\beta_{i,t+1}^{T})^{-1}}\|,
\end{aligned}
\end{equation}
and we will bound the above terms separately.

Before we proceed, we will define some probabilistic events. Let $t$ satisfy $t\bmod{\zeta}=0$ and $t\leq S$. Fix $\hat{\delta} >0$. With the replacement of $\delta$ by $\hat{\delta}$, let $t\geq \max({t_{1},t_{2},t_{3}})$ defined in Lemma \ref{local x} and Lemma \ref{lemma:local noise}. Let $E_{1}$ be the event such that the event in Lemma \ref{local x} occurs for all $i\in\mathcal{V}$ at time step $t$, i.e., 
\begin{equation*}  
\begin{aligned}
E_{1}&\triangleq \bigcap_{i=1}^{m} \biggl\{\{\|\sum_{j=1}^{t}x_{i,j}x_{i,j}^{*}\|\leq t(\frac{19}{8}\sigma_{x}^{2}+\hat{\mu}^{2})\} \cap\\
&\{\lambda_{min}(\sum_{j=1}^{t}x_{i,j}x_{i,j}^{*})\geq \frac{t\sigma_{x}^{2}}{8}\lambda_{min}(I_{n}+\bar{\mu}_{i,t})\}\biggr\},
\end{aligned}
\end{equation*}
where $\bar{\mu}_{i,t}=\frac{4}{t\sigma_{x}^{2}}\sum_{j=1}^{t}\mu_{i,j}\mu_{i,j}^{*}$. Similarly, let $E_{2}$ be the event such that the event in Lemma \ref{lemma:local noise} occurs for all $i\in\mathcal{V}$ at time step $t$, i.e., 
\begin{equation*}  
\begin{aligned}
&E_{2}\triangleq \bigcap_{i=1}^{m} \biggl \{\|\sum_{j=1}^{t}\eta_{i,j}x_{i,j}^{*}\|\leq\\
&\sqrt{t}\sigma_{\eta}(4\sigma_{x}\sqrt{(n+l)\log{\frac{9}{\hat{\delta}}}}+
\hat{\mu}(\sqrt{2(l+n)}+\sqrt{2\log\frac{2}{\hat{\delta}}}))\biggr\}.
\end{aligned}
\end{equation*}
Applying a union bound over all $i\in\mathcal{V}$, we have
\begin{equation}  \label{E3}
\begin{aligned}
E_{3}\triangleq E_{1}\cap E_{2}
\end{aligned}
\end{equation}
occurs with probability at least $1-4m\hat{\delta}$.
\begin{proposition} \label{a}
Conditioning on event $E_{3}$ in \eqref{E3}, we have
\begin{equation*}  
\begin{aligned}
&\|\bar{\alpha}_{t+1}\|\leq t c_{1}+\sqrt{t}c_{2}, 
\end{aligned}
\end{equation*}
where $c_{1}\triangleq \|\Theta\|(\frac{19}{8}\sigma_{x}^{2}+\hat{\mu}^{2})$,  $c_{2}\triangleq \sigma_{\eta}(4\sigma_{x}\sqrt{(n+l)\log{\frac{9}{\hat{\delta}}}}+\hat{\mu}(\sqrt{2(l+n)}+\sqrt{2\log\frac{2}{\hat{\delta}}}))$.
\end{proposition}
\begin{proof}
We have
\begin{equation*}  
\begin{aligned}
\|\bar{\alpha}_{t+1}\|&=\|\frac{1}{m}\sum_{i=1}^{m}\alpha_{i,t+1}\|=\|\frac{1}{m}\sum_{i=1}^{m}\sum_{j=1}^{t}y_{i,j}x_{i,j}^{*}\|\\
&=\|\frac{1}{m}\sum_{i=1}^{m}\sum_{j=1}^{t}(\Theta x_{i,j}+\eta_{i,j})x_{i,j}^{*}\| \\
&\leq\frac{1}{m}(\|\Theta\|\|\sum_{i=1}^{m}\sum_{j=1}^{t}x_{i,j}x_{i,j}^{*}\|+\|\sum_{i=1}^{m}\sum_{j=1}^{t}\eta_{i,j}x_{i,j}^{*}\|)\\
&\leq t\|\Theta\|(\frac{19}{8}\sigma_{x}^{2}+\hat{\mu}^{2})+\sqrt{t}\sigma_{\eta}(4\sigma_{x}\sqrt{(n+l)\log{\frac{9}{\delta}}}\\
&\quad\quad\quad\quad\quad\quad\quad+\hat{\mu}(\sqrt{2(l+n)}+\sqrt{2\log\frac{2}{\delta}})),
\end{aligned}
\end{equation*}
where the last inequality is due to event $E_{3}$.
\end{proof}

We will use the following result from \cite{diaconis1991geometric}.
\begin{lemma} \label{lemma:doubly stochastic}
Consider a weight matrix $W\in\mathbb{R}^{m \times m}$ that satisfies Assumption \ref{ass:topology}. The following inequality holds:
\begin{equation*}
\begin{aligned}
\max_{i\in\{1,\cdots,m\}} \sum_{j\in \{1,\cdots,m\}}\left|W^{T}(i,j)-\frac{1}{m}\right|\leq \sqrt{m}(\rho(W))^{T},
\end{aligned}
\end{equation*}
where $\rho(W)=\max\{\lambda_{2}(W),-\lambda_{m}(W)\}$.
\end{lemma}

\begin{proposition} \label{ea}
Let Assumption \ref{ass:topology} hold. Conditioning on event $E_{3}$ in \eqref{E3}, for all $i\in \mathcal{V}$, we have
\begin{equation*}  
\begin{aligned}
&\|\epsilon_{\bar{\alpha}_{i,t+1}^{T}}\|\leq m^{\frac{3}{2}}\sqrt{l}(\rho(W))^{T}(t c_{1}+\sqrt{t}c_{2}),
\end{aligned}
\end{equation*}
where $c_{1}$ and $c_{2}$ are defined in Proposition \ref{a}.
\end{proposition}
\begin{proof}
Define $A_{t+1}^{T}=\begin{bmatrix}\alpha_{1,t+1}^{T}, \cdots, \alpha_{m,t+1}^{T}\end{bmatrix}$ and $\bar{A}_{t+1}=\begin{bmatrix}\bar{\alpha}_{t+1}, \cdots, \bar{\alpha}_{t+1}\end{bmatrix}$. For all $i\in\mathcal{V}$, we have
\begin{equation}  
\begin{aligned}
&\|\epsilon_{\bar{\alpha}_{i,t+1}^{T}}\|\leq \|A_{t+1}^{T}-\bar{A}_{t+1}\|, 
\end{aligned}
\end{equation}
since $\epsilon_{\bar{\alpha}_{i,t+1}^{T}}$ is a submatrix of the matrix $A_{t+1}^{T}-\bar{A}_{t+1}$.
Now, let $W\in \mathbb{R}^{m\times m}$ be the weight matrix associated with the communication graph, $\bar{W} \in  \mathbb{R}^{m\times m}$ be a matrix with all components equal to $\frac{1}{m}$, and
\begin{equation*}  
\begin{aligned}
\hat{\alpha}^{i,j}_{t+1}=\begin{bmatrix} \alpha_{1,t+1}(i,j)& \alpha_{2,t+1}(i,j)& \cdots &\alpha_{m,t+1}(i,j) \end{bmatrix}^{*}.
\end{aligned}
\end{equation*} We have
\begin{equation}\label{to use 1}
\begin{aligned} 
&\|A_{t+1}^{T}-\bar{A}_{t+1}\|\leq \|A_{t+1}^{T}-\bar{A}_{t+1}\|_{F}\\
&=\|\vect(A_{t+1}^{T}-\bar{A}_{t+1})\|\\
&=\left\|
\begin{bmatrix}
W^{T}-\bar{W}&0&\cdots&0\\
0&W^{T}-\bar{W}&\cdots&0\\
\vdots&\vdots&\ddots&\vdots&\\
0&0&\cdots& W^{T}-\bar{W}
\end{bmatrix}
\begin{bmatrix}
\hat{\alpha}^{1,1}_{t+1}\\
\hat{\alpha}^{1,2}_{t+1}\\
\vdots\\
\hat{\alpha}^{l,n}_{t+1}
\end{bmatrix}\right\|\\
&\leq\|W^{T}-\bar{W}\|
\left\|\begin{bmatrix}
\alpha_{1,t+1}&\alpha_{2,t+1}&\cdots&\alpha_{m,t+1}
\end{bmatrix}\right\|_{F}.
\end{aligned}
\end{equation}

Now applying Lemma \ref{lemma:doubly stochastic}, we have  
\begin{equation}
\begin{aligned} \label{to use 2}
\|W^{T}-\bar{W}\|&=\|(W^{T})^{*}-\bar{W}^{*}\|\leq \sqrt{m}\|(W^{T})^{*}-\bar{W}^{*}\|_{1}\\
&\leq m (\rho(W))^{T}.
\end{aligned}
\end{equation}
Further, since the rank of the matrix $\begin{bmatrix}
\alpha_{1,t+1}&\alpha_{2,t+1}&\cdots&\alpha_{m,t+1}
\end{bmatrix}$ is at most $l$, conditioning on event $E_{3}$, we have
\begin{equation} \label{to use 3}
\begin{aligned}
&\left\|\begin{bmatrix}
\alpha_{1,t+1}&\alpha_{2,t+1}&\cdots&\alpha_{m,t+1}
\end{bmatrix}\right\|_{F} \\
&\leq\sqrt{l}
\left\|\begin{bmatrix}
\alpha_{1,t+1}&\alpha_{2,t+1}&\cdots&\alpha_{m,t+1}
\end{bmatrix}\right\|\\
&\leq\sqrt{ml} \max_{i\in\{1,\cdots,m\}}(\|\alpha_{i,t+1}\|) \\
&\leq \sqrt{ml} \max_{i\in\{1,\cdots,m\}}(\|\Theta\|\|\sum_{j=1}^{t}x_{i,j}x_{i,j}^{*}\|+\|\sum_{j=1}^{t}\eta_{i,j}x_{i,j}^{*}\|)\\
&\leq\sqrt{ml}(t\|\Theta\|(\frac{19}{8}\sigma_{x}^{2}+\hat{\mu}^{2})\\
&+\sqrt{t}\sigma_{\eta}(4\sigma_{x}\sqrt{(n+l)\log{\frac{9}{\hat{\delta}}}}+\hat{\mu}(\sqrt{2(l+n)}+\sqrt{2\log\frac{2}{\hat{\delta}}})).
\end{aligned}
\end{equation}

The result follows by substituting \eqref{to use 2} and \eqref{to use 3} into \eqref{to use 1}.
\end{proof}

\begin{proposition} \label{b}
Conditioning on event $E_{3}$ in \eqref{E3}, we have
\begin{equation*}  
\begin{aligned}
&\|\bar{\beta}_{t+1}^{-1}\|\leq \frac{8}{\sigma_{x}^{2}t}.
\end{aligned}
\end{equation*}
\end{proposition}
\begin{proof}
Conditioning on event $E_{3}$ in \eqref{E3}, we have
\begin{equation*}  
\begin{aligned}
&\bar{\beta}_{t+1}=\frac{1}{m}\sum_{i=1}^{m}\beta_{i,t+1} \succeq \frac{\sigma_{x}^{2}}{8}t I_{n}.
\end{aligned}
\end{equation*}
Taking the inverse we get the desired result.
\end{proof}

\begin{proposition} \label{eb}
Let Assumption \ref{ass:topology} hold. Conditioning on event $E_{3}$ in \eqref{E3}, for all $i\in \mathcal{V}$, we have
\begin{equation*}  
\begin{aligned}
&\|\epsilon_{(\bar{\beta}_{i,t+1}^{T})^{-1}}\|\leq (\rho(W))^{T}\frac{c_{3}}{t},
\end{aligned}
\end{equation*}
where $c_{3}\triangleq \frac{152m^{\frac{3}{2}}\sqrt{5n}}{\sigma_{x}^{2}}+\frac{64m^{\frac{3}{2}}\sqrt{5n}}{\sigma_{x}^{4}}\hat{\mu}^{2}$.
\end{proposition}
\begin{proof}
Denote $\epsilon_{\beta_{i,t+1}^{T}}=\beta_{i,t+1}^{T}-\bar{\beta}_{t+1}$. Letting $W\in \mathbb{R}^{m\times m}$ be the weight matrix associated with the communication graph, $\bar{W} \in  \mathbb{R}^{m\times m}$ be the matrix with all components equal to $\frac{1}{m}$.  Following the same procedure as in the proof of Proposition \ref{ea}, for all $i\in \mathcal{V}$, we have
\begin{equation} \label{tmp1}
\begin{aligned}
&\|\epsilon_{\bar{\beta}_{i,t+1}^{T}}\|\\
&\leq\|W^{T}-\bar{W}\|\left\|\begin{bmatrix}
\beta_{1,t+1}&\beta_{2,t+1}&\cdots&\beta_{m,t+1}
\end{bmatrix}\right\|_{F}\\
&\leq m^{\frac{3}{2}}\sqrt{n} (\rho(W))^{T} \max_{i\in\{1,\cdots,m\}}(\|\beta_{i,t+1}\|)\\
&\leq m^{\frac{3}{2}}\sqrt{n} (\rho(W))^{T}t(\frac{19}{8}\sigma_{x}^{2}+\hat{\mu}^{2}),
\end{aligned}
\end{equation}
where the last inequality is due to event $E_{3}$.

Further, for all $i\in \mathcal{V}$, we have
\begin{equation}   \label{tmp2}
\begin{aligned}
&\|\epsilon_{(\bar{\beta}_{i,t+1}^{T})^{-1}}\|=\|(\beta_{i,t+1}^{T})^{-1}-\bar{\beta}_{t+1}^{-1}\| \\
&\leq \sqrt{5}\max\{\frac{1}{\sigma_{min}^{2}(\beta_{i,t+1}^{T})},\frac{1}{\sigma_{min}^{2}(\bar{\beta}_{t+1})}\} \|\epsilon_{\bar{\beta}_{i,t+1}^{T}}\|,
\end{aligned}
\end{equation}
where the inequality comes from \cite{stewart1977perturbation}.

Conditioning on event $E_{3}$ and noting that $\beta_{i,t+1}^{T}=\sum_{i=1}^{m}q_{i}\beta_{i,t+1}$ for some weights $0\leq q_{i}\leq 1$ and $\sum_{i=1}^{m}q_{i}=1$, we have
\begin{equation}  
\begin{aligned}
&\sigma_{min}(\beta_{i,t+1}^{T})=\sigma_{min}(\sum_{i=1}^{m}q_{i}\beta_{i,t+1})\geq \frac{\sigma_{x}^{2}}{8}t.
\end{aligned}
\end{equation}


Consequently, substituting the above inequality and Proposition \ref{b} into \eqref{tmp2}, and combining with \eqref{tmp1}, we obtain
\begin{equation*}  
\begin{aligned}
\|\epsilon_{(\bar{\beta}_{i,t+1}^{T})^{-1}}\|&\leq \frac{64\sqrt{5}}{\sigma_{x}^{4}t^{2}} \|\epsilon_{\bar{\beta}_{i,t+1}^{T}}\|\\
&\leq\frac{64\sqrt{5}}{\sigma_{x}^{4}t^{2}}m^{\frac{3}{2}}\sqrt{n} (\rho(W))^{T}t(\frac{19}{8}\sigma_{x}^{2}+\hat{\mu}^{2}).
\end{aligned}
\end{equation*}
\end{proof}

Now we are ready to bound the local estimation error after communication. 
\begin{theorem} \label{Thm:cimmunicated error}
Let Assumptions \ref{ass:distribution} and \ref{ass:topology} hold. Fix $\hat{\delta}>0$ and let $t\geq \max({t_{1},t_{2},t_{3}})$, where $t_{1}= 8n+16\log\frac{2}{\hat{\delta}}, t_{2}=\left(\frac{16\hat{\mu}(\sqrt{4n}+\sqrt{2\log\frac{2}{\hat{\delta}}})}{\sigma_{x}}\right)^{2}, t_{3}=2(n+l)\log\frac{1}{\hat{\delta}}$. For $t\bmod{\zeta}=0$ and $t\leq S$, fix $\delta>0$ and denote $\bar{\mu}_{t}=\frac{4}{mt\sigma_{x}^{2}}\sum_{i=1}^{m}\sum_{j=1}^{t}\mu_{i,j}\mu_{i,j}^{*}$.  We have with probability at least $1-4m\hat{\delta}-4\delta$,
\begin{equation}\label{Bound TH}
\begin{aligned}
\|\bar{\Theta}_{i,t+1}-\Theta\|&\leq \underbrace{(\rho(W))^{T} C_{0}}_\text{Error due to network convergence} \\
&+\underbrace{\frac{C_{1}}{\sqrt{mt}\sigma_{x}^{2}\lambda_{min}(I_{n}+\bar{\mu}_{t})}}_\text{Error due to noise},
\end{aligned}
\end{equation}
for all $i\in \mathcal{V}$, where $C_{0}=c_{3}(c_{1}+t^{-1/2}c_{2})+\frac{8m^{\frac{3}{2}}\sqrt{l}(c_{1}+t^{-1/2}c_{2})}{\sigma_{x}^{2}}+(\rho(W))^{T}m^{\frac{3}{2}}\sqrt{l}c_{3}(c_{1}+t^{-1/2}c_{2})$, and $C_{1}, c_{1},c_{2},c_{3}$ are defined in Theorem \ref{Thm:local bound}, Proposition \ref{a} and Proposition \ref{eb}.
\end{theorem}
\begin{proof}
Recall the decomposition of error from \eqref{Communicated error} and \eqref{decomposition of error}. Note that the event $E_{3}$ in \eqref{E3} occurs with probability at least $1-4m\hat{\delta}$ when $t\geq \max({t_{1},t_{2},t_{3}})$. Combine event $E_{3}$ and the event in Theorem \ref{Thm:global error} using a union bound. Applying Propositions \ref{a}, \ref{ea}, \ref{b} and \ref{eb}, we get the desired result.
\end{proof}
\begin{Remark}
Theorem \ref{Thm:cimmunicated error} demonstrates a trade-off between estimation error and communication costs. By choosing $\hat{\delta}$ small, as $T$ tends to infinity, the first term in the bound tends to zero, and agent $i\in \mathcal{V}$ can almost recover the same performance guarantee as if it had access to all samples across the network up to time step $t$ (note that the second term in the error bound reduces the local estimation error bound in Theorem \ref{Thm:local bound} by approximately $\frac{1}{\sqrt{m}}$). The speed at which the first term goes to zero depends on the network topology. Further, this result implies that by choosing $T$ large such that the first term in the bound is small, communication becomes less important as $t$ increases (i.e., as each agent keeps collecting samples), since the second term goes to zero more slowly. Consequently, the improvements of the new local estimate after communication over the old estimates $\bar{\Theta}_{i,t+1}$ and $\hat{\Theta}_{i,t+1}$ will become smaller. In the next section, we discuss how to choose those user specified parameters to balance the trade-off between estimation error and communication costs, leveraging the above observation. 
\end{Remark}

\section{Determining the Communication Period, the Stopping Time, and the number of Communication Steps}\label{discussion}
In short, the communication can be stopped when the minimum between the largest local error bound (over agents) in Theorem \ref{Thm:local bound} and the error bound in Theorem \ref{Thm:cimmunicated error} is less than some pre-specified threshold value $\epsilon\in \mathbb{R}_{>0}$. To achieve that, one needs to first specify the communication period $\zeta$. Note that larger $\zeta$ corresponds to sparser communication. Further, one can specify how much error at most due to network convergence in Theorem \ref{Thm:cimmunicated error} (first term in \eqref{Bound TH}) can be tolerated, denoted as $\epsilon_{N}\in\mathbb{R}_{>0}$. Smaller $\epsilon_{N}$ would require more communication steps. Based on that, one can compute the number of communication steps $T$ that makes the error due to network convergence in Theorem \ref{Thm:cimmunicated error} always less than $\epsilon_{N}$. Consequently, one can then evaluate the bounds in Theorem \ref{Thm:local bound} and Theorem \ref{Thm:cimmunicated error} and determine the stopping time $S$. Note that the bounds in Theorem \ref{Thm:local bound} and Theorem \ref{Thm:cimmunicated error} involve parameters that may be unknown in practice. However, it suffices to replace $\bar{\mu}_{t},\bar{\mu}_{i,t}$ by $0$, and $\sigma_{x},\sigma_{\eta}, \hat{\mu}, \|\Theta\|$ by their corresponding estimated upper/lower bounds. 

Although communication could still help to reduce estimation error after $t>S$, even infinite communication steps can only allow each agent to recover the same estimation error bound as if it had access to the global dataset, under which the reduction of error could be negligible in practice when $\|\hat{\Theta}_{i,t+1}-\Theta\|$ or $\|\bar{\Theta}_{i,t+1}-\Theta\|$ is already small enough. Consequently, it might be preferable for these agents to start updating purely based on local data, considering the communication costs. We will illustrate this idea in the next section empirically.
\section{NUMERICAL EXPERIMENT}
In this example, we consider a network of $m=6$ agents trying to learn model \eqref{eq:True system}, where
\begin{equation*} 
\begin{aligned}
\Theta=\begin{bmatrix}
1.6&0.3\\
0.8&0.3\\
\end{bmatrix},
\sigma_{x}=3, \sigma_{\eta}=1,
\end{aligned}
\end{equation*}
and $\mu_{i,t}=0$ for all $i$ and $t$. The weight matrix associated with the communication graph is
\begin{equation*} 
\begin{aligned}
W=\begin{bmatrix}
1/3&1/3&0&0&0&1/3\\
1/3&1/3&1/3&0&0&0\\
0&1/3&1/3&1/3&0&0\\
0&0&1/3&1/3&1/3&0\\
0&0&0&1/3&1/3&1/3\\
1/3&0&0&0&1/3&1/3\\
\end{bmatrix}.
\end{aligned}
\end{equation*}
We set $\zeta=20$ and assume that all parameters in Theorem \ref{Thm:local bound} and Theorem \ref{Thm:cimmunicated error} are known for simplicity. The number of communication steps is set to $T=38$,  which is computed based on the guidelines suggested in Section \ref{discussion} such that the error due to network convergence in Theorem \ref{Thm:cimmunicated error} is always less than $0.01$ (using $\hat{\delta}=0.001$). The communication is stopped when the smallest error bound between the one in Theorem \ref{Thm:local bound} (using $\delta=0.05$) and the one in Theorem \ref{Thm:cimmunicated error} (using $\delta=0.05, \hat{\delta}=0.001$) is less than $0.5$, which leads to $S=1620$. We plot the average (over agents) local estimation error without communication $\|\hat{\Theta}_{i,t+1}-\Theta\|$, the average local estimation error after communication $\|\bar{\Theta}_{i,t+1}-\Theta\|$, and the global estimation error $\|\hat{\Theta}_{t+1}-\Theta\|$. All results are averaged over 10 independent runs.

As expected, the error $\|\bar{\Theta}_{i,t+1}-\Theta\|$ is almost the same as $\|\hat{\Theta}_{t+1}-\Theta\|$ when communication happens. Further, the error $\|\bar{\Theta}_{i,t+1}-\Theta\|$ decreases relatively rapidly, and is much smaller than $\|\hat{\Theta}_{i,t+1}-\Theta\|$ at the beginning. However, the error $\|\bar{\Theta}_{i,t+1}-\Theta\|$ decreases more slowly, and its improvement over $\|\hat{\Theta}_{i,t+1}-\Theta\|$ becomes smaller, as each agent gathers more samples. Although the communication is stopped at $t=1620$, leveraging the global dataset has only marginal improvements over the estimates $\hat{\Theta}_{i,t+1},\bar{\Theta}_{i,t+1}$ after $t=1620$, implying communication becomes less important, which confirms our observation in Theorem \ref{Thm:cimmunicated error}. On the other hand, although the minimum between  $\|\hat{\Theta}_{i,t+1}-\Theta\|$ and $\|\bar{\Theta}_{i,t+1}-\Theta\|$ is less than $0.5$ after $t=1620$, the simulation also implies that our finite time bound is conservative. It is of interest to develop tighter bounds in future work.  

\begin{figure}[ht]
\centering
\includegraphics[width=\linewidth]{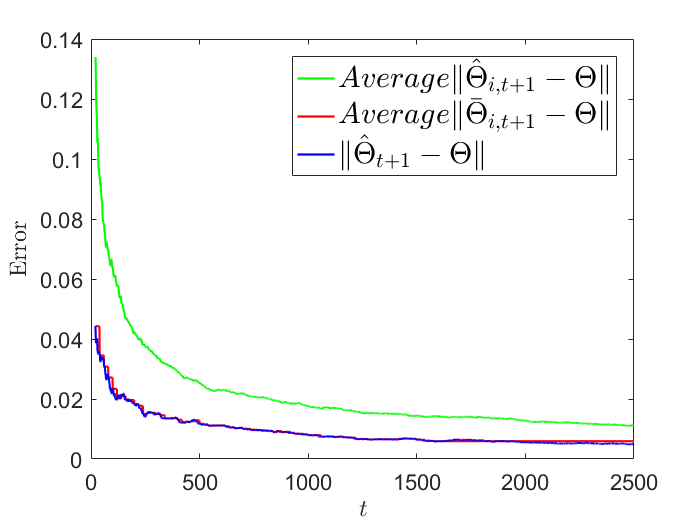}
\caption{Average $\|\hat{\Theta}_{i,t+1}-\Theta\|$, average $\|\bar{\Theta}_{i,t+1}-\Theta\|$, and $\|\hat{\Theta}_{t+1}-\Theta\|$. The communication is stopped after $t=1620$.}
\label{fig: numerical experiment}
\end{figure}

\section{Conclusion and future work} \label{sec: conclusion}
In this paper, we proposed an online distributed parameter estimation algorithm with finite time performance guarantees. Our results demonstrate a trade-off between estimation error and communication costs, and we show that one can leverage the error bounds to determine a time at which the communication can be stopped. We believe our results can be extended to more general graph conditions, e.g., leveraging \cite[Proposition~1]{nedic2009distributed}. Future work could focus on developing similar bounds by using gradient-based methods.




\bibliographystyle{IEEEtran}
\bibliography{main}
\end{document}